\documentclass[runningheads,a4paper]{llncs}

\begingroup
  \def\x{\endgroup\ExecuteOptions{dvipdfm}}%
  \ifx\pdfoutput\undefined
  \else
    \ifx\pdfoutput\relax
    \else
      \ifcase\pdfoutput
        
      \else
        \def\x{\endgroup\ExecuteOptions{pdftex}}%
	
      \fi
    \fi
  \fi
\x

\usepackage{tikz}  

\usepackage{graphicx}
\usepackage{amsfonts}
\usepackage{amsmath}
\usepackage{amssymb}
\usepackage{mathtools}
\usepackage{longtable}
\usepackage{ulem}
\usepackage{enumerate}
\usepackage{url}
\usepackage{ulem}
\usepackage[T1]{fontenc}
\usepackage{times}
\usepackage{graphicx}
\usepackage[bookmarks=true,colorlinks=true]{hyperref}    
\usepackage{array}
\usepackage[boxed,vlined,linesnumbered,noend]{algorithm2e}
\usetikzlibrary{matrix,arrows,decorations.pathmorphing, decorations.pathreplacing, calc, shapes,patterns, fit, positioning}
\usepackage{subfig} 
\usepackage{float}

\newtheorem{fact}{Fact}
\newtheorem{df}{Definition}
\newtheorem{exmp}{Example}
\SetKwInOut{Input}{Input}
\SetKwInOut{RestInput}{}
\SetKwInOut{Output}{Output}
\SetKwInOut{RestOutput}{}
\SetKwInOut{Actors}{Actors}
\SetAlFnt{\small}
\newcommand{\keywords}[1]{\par\addvspace\baselineskip
\noindent\keywordname\enspace\ignorespaces#1}

\newcommand{\anonymous}[2]{#1} 

\newcommand{\compressaffil}[2]{#1} 
                            
\renewcommand{\S}{{\mathcal{S}}}
\newcommand{\C}{{\mathcal{C}}}
\renewcommand{\P}{\mathbb{P}}
\newcommand{\T}{{\mathcal{T}}}
\newcommand{\ep}{\varepsilon}

\begin{document}

\title{What Do Our Choices Say About Our Preferences?}
\titlerunning{What Do Our Choices Say About Our Preferences?}


\author{%
\anonymous{
\mbox{Krzysztof Grining\inst{1}}, \mbox{Marek Klonowski\inst{2}},  
\mbox{Ma{\l}gorzata Sulkowska\inst{1}}}
{Of anonymous authors}
}

\institute{
\anonymous{ \compressaffil{Department of Fundamentals of Computer Science \\
		Wroc{\l}aw University of Science and Technology \\ \email{\{firstname.secondname\}@pwr.edu.pl}  \vspace{5pt} \and 
		Department of Artificial Intelligence \\
		Wroc{\l}aw University of Science and Technology \\ \email{\{firstname.secondname\}@pwr.edu.pl}
	}{}
}{...
}}

\maketitle

\begin{abstract}

Taking online decisions is a part of everyday life. Think of buying a house, parking a car or taking part in an auction. We often take those decisions publicly, which may breach our privacy - a party observing our choices may learn a lot about our preferences. In this paper we investigate the online stopping algorithms from the privacy preserving perspective, using a mathematically rigorous differential privacy notion. 

In differentially private algorithms there is usually an issue of balancing the privacy and utility. In this regime, in most cases, having both optimality and a high level of privacy at the same time is impossible. We propose a natural mechanism to achieve a controllable trade-off, quantified by a parameter, between the accuracy of the online algorithm and its privacy. Depending on the parameter, our mechanism can be optimal with weaker differential privacy or suboptimal, yet more privacy-preserving. We conduct a detailed accuracy and privacy analysis of our mechanism applied to the optimal algorithm for the classical secretary problem. Thereby the classical notions from two distinct areas - optimal stopping and differential privacy - meet for the first time. 



\end{abstract}

 \section{Introduction}\label{sect:intro}
We make online decisions every day - buying a house, trading stock options or parking a car. The choices we make are mostly based on our knowledge and experience. Those decisions are most often publicly visible and it raises concern about the internal information on which the choice was based. Namely, our choices may somewhat leak that information, which we consider as sensitive. Someone could observe our choices and deduce our preferences or domain knowledge of our algorithmic trading company. In this paper we consider the security of such information in case of optimal stopping algorithms.

The optimal stopping algorithms are widely known and thoroughly researched. One of the most classical models in this area is the \textit{secretary problem}.
There is a set of $n$ linearly ordered candidates and a \textit{selector} would like to choose the best candidate from that set. The caveat is that the candidates appear online one by one in some random order. I.e., the selector cannot see all of them at once and simply pick the best one. They have only one choice and have to make it online. 
The decision must be based on an incomplete knowledge gathered by comparing the current candidate with the previously seen ones, and it is irreversible, as the candidate cannot be hired after he or she was rejected.  The problem was popularized  in 1960 by Martin Gardner in \textit{Scientific American} column under the name of \textit{googol game} and the optimal solution was first written down by Lindley in \cite{Lindley}. The optimal solution to that problem gives the asymptotic probability of $1/e$ for choosing the best candidate (see \cite{Lindley}). 
For a historical overview of the secretary problem consult Ferguson's survey \cite{Ferguson}. Many generalizations of this classical version were considered later on. Stadje was the first one who replaced a linear order of candidates with a partial order (poset) and aimed at choosing any maximal element~\cite{Stadje}. An account of research considering threshold strategies for posets was given by Gnedin in \cite{Gnedin}. Optimal strategies for a particular posets as well as universal algorithms for the whole families of posets have been featured in \cite{bin_tree,Garrod,univ_poset,Kozik,Freij_Wastlund}. The problem was investigated even on much more general structures, like directed graphs \cite{Kubicki_Morayne_2005,Przykucki_Sulkowska_2010,Sulkowska_universal_2012,Goddard_Kubiccy_2013,Grzesik_Morayne_Sulkowska_2015,Benevides_Sulkowska_2017} or matroids \cite{matroid_secretary,matroid_STOC2016,matroid_2018}.
Other interesting extensions consider different payoff functions. A natural reformulation of the classical case is to aim at minimizing the expected rank of the candidate instead of selecting just the best one. This model was introduced by Lindley in \cite{Lindley} but fully solved by Chow et al. in \cite{Chow}. 
Even though formulated probably in the mid-1950s, the secretary problem with its generalizations is still vivid and attracts the attention of theoretical computer science community, see the latest results in \cite{SODA20,SODA21,multi_secr_2022,Albers_TCS_2021,Fujii_2023,secretary_management_2024}.


In this paper we investigate the optimal stopping algorithms from a privacy-protection perspective, which, to the best of our knowledge, has not been done so far.  
Assuming that the online choice we make is publicly visible, it may leak some information about our preferences.
Note that if the choice was possible offline, i.e., we could first see all the candidates and then pick the best one, the leakage would be unavoidable, as the chosen candidate would simply be our best one. But when the choice has to be made in an online regime (we pick or reject the candidates on the fly, without the possibility of revisiting the rejected ones) then on one hand the outcome is not perfectly accurate but on the other our preferences are not that visible. 
The optimal stopping algorithms can be used for example in algorithmic trading. One might care whether the visible action on the market (e.g., closing an American option position) leaks something about an internal knowledge.
Hereby we try to answer the question whether the inherent uncertainty of the online stopping algorithms is enough to sufficiently hide our preferences.

Our analysis of privacy is based on \textit{differential privacy} notion commonly considered as the only state-of-the-art approach. Its idea was introduced by Dwork et al. in~\cite{dwork2006calibrating}, however its precise formulation in the widely used form  appeared in~\cite{Dwork06}. Differential privacy is mathematically rigorous and formally provable contrary to the previous anonymity-derived privacy definitions (for a comprehensive study check~\cite{DworkAlgo} and references therein). Informally, the idea behind differential privacy is as follows: for two "similar" inputs, a differentially private mechanism should provide a response chosen from very similar distributions. In effect, judging by the output of the mechanism one cannot say if a given individual was taken into account for producing a given output, as it cannot distinguish two outputs produced from two data sets differing with only one user. This is mostly done by adding an auxiliary randomness (e.g., a carefully calibrated noise) to data. The vast majority of the fundamental papers (e.g.~\cite{dwork2006our,dwork2009differential,dwork2010differential}) consider centralized model as is assumed in our contribution. In this model a trusted party (called the \textit{curator}) holds a database. 
He is entitled to gather and process all participants' data. The curator also releases the computed statistics to a wider (possibly untrusted) audience. Such an approach to the privacy-preserving protocols can be used to give a formal guarantee for privacy resilience to any form of post-processing.
Analysis of the protocols based on differential privacy is usually technically much more involved compared to previous approaches, but they give immunity against various linkage attacks (see e.g.  \cite{narayanan2009anonymizing,narayanan2010myths}).


\subsection{Results}
We can think of two extremes while making an online decision - one is to act according to the known, optimal algorithm and the other is to act completely randomly. Obviously, the first approach yields the best possible chance to make a correct pick but leaks internal information while the latter does not leak anything but selects a candidate only at random.
Depending on the nature of the problem and the importance of information to be hidden we have to make a trade-off. 

In this paper we present a natural algorithm that is such a trade-off.  It has a steering parameter $p \in [0,1]$: by $p=1$ it is optimal, and by $p=0$ it hides preferences perfectly. We analyze its effectiveness and privacy properties. 
Subsequently, we apply it to the classical secretary problem. 
It turns out that already the optimal secretary algorithm itself ensures some privacy, but rather only in weak metrics. Our approach, based on adding a randomized perturbation, is common in privacy mechanisms. This intuitively obvious algorithm turned out to be non-trivial to analyze. 
This work can be seen as the first step to tackle privacy and information hiding in the optimal stopping algorithms so we focused only on the classical model. \\

\textbf{Main results}
\begin{itemize}
\item We introduce the definition of a \textit{differentially private stopping time} (see Definition~\ref{df:priv_stop} in Section \ref{sect:model}) with a metric defined on a set of orderings representing the selector's preferences (see Definition~\ref{df:metric} in Section \ref{sect:model}).
\item In Theorem~\ref{thm:l_inf} we show a fundamental lower bound for the possibility of constructing privacy preserving algorithms that are close to optimal.
\item We propose a natural mechanism transforming any optimal stopping time for a linearly ordered set of candidates into an algorithm preserving accuracy at a controllable level and having better information-hiding properties. The formulation of the algorithm is presented in Section \ref{sect:results}. General results concerning its accuracy and differential privacy are stated in Section~\ref{sec:general_pmix} (see Fact~\ref{fact:mix_efficiency} and Theorem~\ref{thm:optimal}).
\item We conduct a detailed analysis of our mechanism applied to the optimal stopping algorithm for the classical secretary problem. The results on its accuracy and differential privacy are stated in Section~\ref{sec:secretary_pmix} (see Fact~\ref{fact:efficiency_pmix} and Theorem~\ref{thm:main_eps}). In particular, we obtain the privacy properties of the optimal algorithm for the secretary problem.
\end{itemize}
In order to obtain the privacy results for our mechanism applied to the secretary problem we had to prove a series of collateral properties of the optimal algorithm for the secretary problem (see the Appendix).


 \section{Formal model}\label{sect:model}

\subsection{Stopping time}\label{subsect:general_model}

Let $\mathbb{N}$ be the set of natural numbers. For $n \in \mathbb{N}$ by $[n]$ denote the set $\{1,2,\ldots,n\}$ and by $\S_n$ the set of all permutations of elements from $[n]$, 
thus $|\S_n| = n!$. We write $\C = \{C_1, C_2, \ldots, C_n\}$ for the set of $n$ candidates. Whenever we use a symbol $\sigma \in \S_n$ we interpret it as the ordering of candidates with respect to their qualifications, i.e., $\sigma = (\sigma_1, \sigma_2, \ldots, \sigma_n)$ 
refers to the ordering $C_{\sigma_1} > C_{\sigma_2} > \ldots > C_{\sigma_n}$ (in particular, it means that the candidate $C_{\sigma_1}$ has the best and the candidate $C_{\sigma_n}$ has the worst qualifications from the whole group $\C$). These orderings are called the \textit{qualification orderings}.

Whenever we use a symbol $\tau \in \S_n$ we interpret it as the sequence giving the order in which the candidates appear in some online game. Thus $\tau = (\tau_1, \tau_2, \ldots, \tau_n)$ means, in particular, that the candidate $C_{\tau_1}$ appeared as the first one and the candidate $C_{\tau_n}$ as the last one in our online game. These orderings are called the \textit{time orderings}.

Throughout this paper we refer to the following model of an online stopping problem. Fix $\sigma$ (choose a particular qualification ordering on the set of candidates). Note that, in fact, $\sigma$ is a preference which we want to hide. Assume also that $\tau$ is chosen uniformly at random from $\S_n$. The player knows $\sigma$ but he does not know $\tau$. Candidates from $\C$ appear one by one following the order given by $\tau$. At time $t$, i.e., when $t$ candidates appeared, the player observes the qualification order induced by $\{\tau_1, \tau_2, \ldots, \tau_t\}$. That is, he knows the relative ranks of the candidates seen so far but he does not know their total ranks. At each time step he has to decide whether to continue the game and reveal the next element or to stop the game meaning that he selects the element $\tau_t$. If he decides to reveal another element, he is not allowed to come back to the previous steps of the game. His task is to maximize (or ensure relatively high) probability that the selected candidate belongs to some previously defined set (e.g. is maximal in the whole set $\C$).

Formally, we define a probability space $(\S_n, {\mathcal{P}}, \P)$, where ${\mathcal{P}}$ is the set of all subsets of $\S_n$ and the probability measure is defined by $\P[\{\tau\}] = 1/n!$ for any $\tau \in \S_n$. A \textit{stopping time} is a function $M: \S_n \times \S_n \to \{1, 2, \ldots, n\}$ such that its value, say $t = M(\sigma,\tau)$, depends only on information the player gathered up to time $t$, which is the qualification order induced by $\{\tau_1, \tau_2, \ldots, \tau_t\}$. The value can not depend on any future events. We give a strict formal definition of a stopping time below.
\begin{df}
	Let  $(\S_n, {\mathcal{P}}, \P)$ be a probability space, where ${\mathcal{P}}=2^{\S_n}$ and  $\P[\{\tau\}] = 1/n!$ for any $\tau \in \S_n$. Let ${\mathcal{P}}_1 \subseteq {\mathcal{P}}_2 \subseteq \ldots \subseteq {\mathcal{P}}_n \subseteq {\mathcal{P}}$ be a sequence of $\sigma$-algebras (a~filtration). A random variable $M: \S_n \times \S_n \to \{1, 2, \ldots, n\}$ is a stopping time with respect to a filtration $(\mathcal{P}_t)_{t=1}^{n}$ if, truncating $M$ to any $\sigma \in \S_n$ on the first coordinate, we have that  $(M|_{\sigma})^{-1}(t) \in \mathcal{P}_t$ for all $t \leq n$.
\end{df}
In our case the sets $A$ in $\mathcal{P}_t$ are those with the following property. Fix $\sigma \in \S_n$. If $\tau = (\tau_1, \tau_2, \ldots, \tau_n) \in A$ then for every $\tilde{\tau} \in \S_n$ such that the orders of candidates induced by $\tau$ and $\tilde{\tau}$ are identical up to time $t$ we have $\tilde{\tau} \in A$.

The expression $M(\sigma,\tau)=t$ means that the algorithm $M$ stopped at time $t$, thus selected the candidate $C_{\tau_t}$. Notation for a candidate returned by $M$ while playing on the time ordering $\tau$ is $C_{\tau_{M(\sigma,\tau)}}$. We will often refer to the probability that a candidate returned by $M$ belongs to some subset $S$ of the set of all candidates $\C$, i.e. $\P[C_{\tau_{M(\sigma,\tau)}} \in S]$. In order to clarify the notation we introduce a notion $CM(\sigma,\tau)$ for $C_{\tau_{M(\sigma,\tau)}}$.

\begin{df} \label{df:opt_stop}
Let $\mathcal{M}$ denote the set of all stopping times. We say that $M^{opt}$ is an \textit{optimal stopping time} if
\[
M^{opt} = {\arg\!\max}_{M \in {\mathcal{M}}} {\P[CM(\sigma,\tau) \in D]},
\]
where $D \subseteq \C$ is a set of previously defined candidates, $\sigma \in \S_n$, and $\tau$ is chosen uniformly at random from $\S_n$.
\end{df}

%

We write $f(n) \sim g(n)$ whenever $\lim_{n \to \infty} f(n)/g(n) = 1$.

\subsection{Secretary problem}\label{subsect:secretary}

From now on $M^*$ always denotes the optimal stopping time for the classical \textit{secretary problem}. In the secretary problem the player aims at maximizing the probability of selecting the candidate which is the best from the whole $\C$. That is, for a particular $\sigma \in \S_n$, the set $D$ from Definition \ref{df:opt_stop} is given by $D = \{C_{\sigma_1}\}$. A full solution to this problem (the optimal algorithm and its probability of success), as already mentioned in the introduction, is well known \cite{Lindley,Ferguson}. The asymptotic results are also established. We present them below.

\begin{df}
Let $\sigma \in \S_n$ be the qualification ordering of the elements from $\C$. Let the value $t_n$ be a so-called threshold of the algorithm. The \textit{optimal algorithm $M^*$ for the secretary problem} (the one maximizing $\P[CM(\sigma,\tau) = C_{\sigma_1}]$ over $M \in {\mathcal{M}}$ assuming that $\tau$ is chosen uniformly at random from $\S_n$) is defined as follows. For any $\tau \in \S_n$ we have $M^*(\sigma, \tau) = k$ if and only if
\begin{enumerate}[(1)]
	\item $k > t_n-1$ and
	\item $\tau_k$ is the maximal element in the qualification ordering induced by $\{\tau_1, \tau_2, \ldots, \tau_k\}$ and
	\item for $i \in  \{t_n, \ldots, k-1\}$ the element $\tau_i$ is not maximal in the qualification ordering induced by $\{\tau_1, \ldots, \tau_i\}$.
\end{enumerate}
If $\tau$ is such that the above three conditions are never altogether satisfied, then $M^*(\sigma, \tau) = n$.
\end{df}

\begin{fact} \label{fact:threshold}
The threshold $t_n$ of the optimal algorithm $M^*$ for the classical secretary problem with $n$ candidates is defined as the smallest integer $t$ for which
\[
\frac{1}{t} + \frac{1}{t+1} + \ldots + \frac{1}{n-1} \leqslant 1.
\]
We have $t_n \sim n/e$ (consult \cite{Lindley}).
\end{fact}

\begin{fact} \label{fact:kth_cand}
Fix $\sigma \in \S_n$ and assume that $\tau$ is chosen uniformly at random from $\S_n$. The probability that the optimal algorithm $M^*$ selects the $k\textsuperscript{th}$ best candidate is given by
\[
\P[CM^*(\sigma, \tau) = C_{\sigma_k}] = \frac{t_n-1}{n} \left( \sum_{i=t_n}^{n-k+1}\frac{\binom{n-k}{i-1}}{\binom{n-1}{i-1}}\frac{1}{i-1} + \frac{1}{n-1} \right),
\]
where $t_n$ is the threshold from Fact \ref{fact:threshold} (consult \cite{Rogerson} or see Theorem \ref{thm:kth_cand} in the Appendix). In particular, the probability that it wins (selects the best candidate) is
\[
\P[CM^*(\sigma, \tau) = C_{\sigma_1}] = \frac{t_n-1}{n}\sum_{i=t_n}^{n} \frac{1}{i-1} \xrightarrow{n \rightarrow \infty} 1/e \approx 0.37.
\]
In general, when $k$ is a constant or a function of $n$ such that $k(n)=o(n)$
\[
\P[CM^*(\sigma, \tau) = C_{\sigma_k}] \sim \frac{1}{e} \sum_{s=k}^{\infty} \frac{1}{s} \left(1-\frac{1}{e}\right)^s
\]
(consult \cite{Rogerson} or see Theorem \ref{thm:kconst_limit} in the Appendix).
\end{fact}

\subsection{Differential privacy}

In this subsection we recall the definition of \textit{differential privacy} and present the privacy model used throughout the paper. For more details about differential privacy see e.g.~\cite{DworkAlgo}.

We assume that there exists a trusted curator who holds data of individuals in the database. 
A \textit{privacy mechanism} is an algorithm, used by the curator, that takes as an input a database and produces an output (a release) using randomization.  
By $\mathcal{X}$ we denote the space of all possible rows in a database (each row consists of data of some individual). The privacy mechanism has a domain $\mathbb{N}^{|\mathcal{X}|}$ representing the set of databases. Thus each database is represented as an $|\mathcal{X}|$-tuple $(n_1,n_2,\ldots,n_{|\mathcal{X}|})$, where $n_k$ is interpreted as the number of rows of kind $k$ in this database. If $x=(n_1,n_2,\ldots,n_{|\mathcal{X}|})$ then $n_1+n_2+\ldots+n_{|\mathcal{X}|}$ is the number of rows in $x$. The goal is to protect the data of every single individual, even if all the users except one collude with the Adversary to breach the privacy of this single, uncorrupted user.   

\begin{df}[Differential Privacy, \cite{DworkAlgo}]\label{dpDef}
Let $\varepsilon \geq 0$ and $\delta \in [0,1]$. A randomized algorithm $M$ with the domain $\mathbb{N}^{|\mathcal{X}|}$ is $(\varepsilon,\delta)$-differentially private, if for all $S \subseteq$ Range($M$) and for all $x, y \in \mathbb{N}^{|\mathcal{X}|}$ such that $\left\Vert x-y \right\Vert_1 \leqslant 1$ the following condition is satisfied:
$$
\P[{M(x) \in S}] \leqslant e^\varepsilon\cdot \P[{M(y) \in S}] + \delta,
$$
where the probability space is over the outcomes of $M$ and $\left\Vert \cdot \right\Vert_1$ denotes the standard $l_1$ norm.
\end{df}
The intuition behind the $(\varepsilon,\delta)$-differential privacy is that if we choose two consecutive databases (that differ exactly on one record), then the mechanism is very likely to return indistinguishable values. Speaking informally, it preserves privacy with high probability, thus the outcome could be distinguishable, thus not privacy-preserving, only with probability at most $\delta$.

In this paper we consider the optimal stopping algorithms and by the database we understand a permutation of the set of $n$ choices (e.g., candidates in the secretary problem). Rather than hiding the participation of a candidate in a stopping game, we wish to hide the preferences of the selector. Indeed it is the preference that is sensitive, not the participation itself. For example, in financial markets the set of candidates is publicly known (say, the prices of a stock). On the other hand, our reaction to that set based on the underlying assumptions and the domain knowledge, which are connected with the preferences permutation, is sensitive and needs protection.

Below we present a differential privacy definition reformulated for our purposes.

\begin{df} \label{df:priv_stop}
Let $\varepsilon \geq 0$ and $\delta \in [0,1]$. Let $(\S_n,d)$ be a metric space. A randomized algorithm $M$ with the domain $\S_n \times \S_n$ is an $(\ep,\delta)$-differentially private stopping time w.r.t. metric $d$ if for all $S \subseteq$ Range($M$) and for all $\sigma, \rho \in \S_n$ such that $d(\sigma,\rho) \leqslant 1$ we have
\[
\P[CM(\sigma, \tau) \in S] \leqslant e^{\ep} \P[CM(\rho, \tau) \in S] + \delta,
\]
where $\tau$ is chosen uniformly at random from $\S_n$.
\end{df}

All the probabilities are calculated with respect to the fact that $\tau$ is chosen uniformly at random from $\S_n$. Note that this definition does not guarantee privacy in all cases (i.e., in the worst-case scenario). Indeed, with some small and controllable probability $\delta$ the information leakage violating privacy is possible. In our case this probability is computed over a random order of candidates that are considered subsequently. Let us stress that in the gross of realistic scenarios the negligible but nonzero $\delta$ probability of privacy failure is indispensable since some input characteristics reveal some extra information. Many examples with interesting interplay between parameters $\varepsilon$ and $\delta$ may be found in \cite{DworkAlgo}. For the example of $(0,0)$-differentially private stopping time regardless of the metric check Definition~\ref{df:blind} and Fact~\ref{fact:uniform_private}. From now on, for short, we sometimes write $(\ep,\delta)$-DP for an $(\ep,\delta)$-differentially private stopping time. 

We introduce a definition of a parametrized metric on $\S_n$. 
\begin{df} \label{df:metric}
For $l \in \{1, 2, \ldots, n-1\}$ let $d_l:\S_n \times \S_n \rightarrow [0,\infty)$. We say that $d_l$ is an \textit{$l$-distance} between the permutations $\sigma, \rho \in \S_n$ if
\[
d_l(\sigma,\rho) = \min \{k: \sigma = \pi_1 \circ \pi_2 \circ \ldots \circ \pi_k \circ \rho\},
\]
where $\pi_i \in \S_n$ and each $\pi_i$ is a transposition of the elements being at most $l$ apart in the permutation $\pi_{i+1} \circ \ldots \circ \pi_k \circ \rho$.
\end{df}


\begin{exmp}
	Let $\sigma=(1,2,3, \ldots, n-1, n)$ and $\rho=(n, 2, 3, \ldots, n-1, 1)$. We have $d_1(\sigma, \rho) = 2n-3$ since we need to make at least $2n-3$ swaps of the neighboring elements in order to obtain $\sigma$ from $\rho$:
	\[
	\sigma = (n~n-1) \ldots  (n~3) \circ (n~2) \circ (1~n) \ldots (1~n-2) \circ (1~n-1) \circ \rho.
	\]
	However, $d_{n-1}(\sigma, \rho) = 1$ since we need just one swap of the elements $n-1$ apart in order to obtain $\sigma$ from $\rho$:
	\[
	\sigma = (1~n) \circ \rho.
	\]
\end{exmp}

\begin{fact}
Let $l \in \{1, 2, \ldots, n-1\}$. It is easy to show that $(S_n,d_l)$ is a metric space.
\end{fact}

Note that $d_{n-1}$ is the strongest metric as it treats two permutations with any pair of swapped elements as neighboring. On the other hand we have $d_1$ -- here permutations are neighboring only if it is a pair with the neighboring elements swapped. 
From the privacy perspective point of view these metrics significantly differ. A privacy mechanism using $d_{n-1}$, intuitively, hides our preferences completely. It should e.g. output a similar outcome even for the qualification orderings $\sigma$ and $\tilde{\sigma}$, where $\tilde{\sigma}$ swaps in $\sigma$ the best candidate with the worst one. In general, $d_l$ metric hides preferences for up to $l$ distance. I.e., say that a picked candidate was in reality $k^{\text{th}}$ on the full preference list. Then from the Adversary perspective he or she could have been between $k-l$ and $k+l$ on the preference list. 

Note that if an algorithm is private in stronger metrics, the Adversary cannot really guess the preferences of the selector, despite knowing the picked candidate. Intuitively, constructing such an algorithm having a high probability of success seems hard to achieve, as then the final choice should say almost nothing about the selector's preferences. Indeed, in the next section we prove that achieving a reasonable differential privacy parameters ($\varepsilon$ and $\delta$) and a constant probability of success in case of metrics $d_l$ such that $l=l(n) \xrightarrow{n \rightarrow \infty} \infty$ is impossible. 

\section{General results}

In this section we present two general results for the  $(\varepsilon,\delta)$-DP stopping times with the metric $d_l$. In the first one we refer to the problem of choosing the best candidate, thus for a fixed $\sigma \in \S_n$ the set $D$ from Definition \ref{df:opt_stop} is given by $D = \{C_{\sigma_1}\}$. This result tells that by $l=l(n) \xrightarrow{n \to \infty} \infty$ it is impossible to construct an algorithm with a constant probability of success having reasonable privacy parameters. The second result constitutes the lower bound for $\varepsilon$ for a given optimal stopping time and $\delta$.

\begin{theorem} \label{thm:l_inf}
Let $\sigma \in \S_n$, and let $\varepsilon \geq 0$ and $\delta \in [0,1]$. Consider a metric space $(\S_n,d_l)$ for $l = l(n)\xrightarrow{n \to \infty} \infty$. 
Let $M$ be a stopping time $M$ such that $\P[CM(\sigma,\tau) = C_{\sigma_i}]>0$ for all $i \in \{1,2,\ldots,n\}$, and $\liminf_{n \to \infty} \left( \P[CM(\sigma,\tau) = C_{\sigma_1}] - \delta \right) = c >0$.
If $M$ is $(\ep,\delta)$-DP w.r.t. $d_l$, then $\P[CM(\sigma,\tau) = \sigma_1] \xrightarrow{n \to \infty} 0 $ or $\ep \xrightarrow{n \to \infty} \infty$.
\end{theorem}
\begin{remark}
Here we assume that $\delta$ is significantly smaller than $\P[CM(\sigma,\tau) = C_{\sigma_1}]$, i.e. that the difference between $\P[CM(\sigma,\tau) = C_{\sigma_1}]$ and $\delta$ tends to some positive constant with $n \to \infty$. Note that in practice one demands $\delta$ to be a very small constant (or even a value tending to $0$ with $n$). Therefore the assumption is reasonable.
\end{remark}
\begin{proof}
Assume by contradiction that $l = l(n)\xrightarrow{n \to \infty} \infty$ and the stopping algorithm $M$ is such that $\P[CM(\sigma,\tau) = C_{\sigma_1}]$ is a positive constant and $M$ is $(\ep,\delta)$-DP w.r.t. $d_l$ with a constant $\ep$. Consider the following qualification orderings, all being at a distance~$1$ from $\sigma$ w.r.t. $d_l$:
\[
\begin{split}
\rho^{(2)} & = (\sigma_1~\sigma_{2}) \circ \sigma = (\sigma_{2}, \sigma_1, \sigma_3, \ldots, \sigma_n),\\
\rho^{(3)} & = (\sigma_1~\sigma_{3}) \circ \sigma = (\sigma_{3}, \sigma_2, \sigma_1, \sigma_4, \ldots, \sigma_n),\\
& \ldots \\
\rho^{(l+1)} & = (\sigma_1~\sigma_{l+1}) \circ \sigma = (\sigma_{l+1}, \sigma_2, \ldots, \sigma_{l}, \sigma_{1}, \sigma_{l+2}, \ldots, \sigma_n).
\end{split}
\]
Let $\P[CM(\sigma,\tau) = C_{\sigma_j}] = q_{j,n}$. Note that for $j=2,3,\ldots,l+1$
\[
\P[CM(\rho^{(j)},\tau) = C_{\sigma_1}] = q_{j,n}.
\]
Since $M$ is $(\ep,\delta)$-DP w.r.t. $d_l$ the following system of $l$ inequalities is satisfied
\[
\ep \geq \ln{\frac{q_{1,n}-\delta}{q_{2,n}}}, \quad \ep \geq \ln{\frac{q_{1,n}-\delta}{q_{3,n}}}, \ldots, \quad \ep \geq \ln{\frac{q_{1,n}-\delta}{q_{l+1,n}}}
\]
(consider $S = \{C_{\sigma_1}\}$ in Definition \ref{df:priv_stop}). We infer that since $q_{1,n}$ is a constant and $\ep$ is a constant, the probabilities $q_{j,n}$ for $j=2,3,\ldots,l+1$ are also constant. Now, since $l = l(n)\xrightarrow{n \to \infty} \infty$ we get $q_{1,n} + q_{2,n} + \ldots + q_{l+1,n} \xrightarrow{n \to \infty} \infty$ which contradicts the fact that $\sum_{j=1}^{n} q_{j,n} = 1$.
\qed
\end{proof}

Therefore throughout the rest of the paper we always assume that $l$ is a constant.

\begin{theorem} \label{thm:gen_eps}
Consider a metric space $(\S_n,d_l)$. Let $\varepsilon \geq 0$, $\delta \in [0,1]$, and let $M$ be a stopping algorithm which is $(\ep,\delta)$-DP w.r.t. $d_l$. For a given $\sigma \in \S_n$ let 
\[
\P[CM(\sigma,\tau) = C_{\sigma_i}] = q_{i,n} \in (0,1), \quad i=1,2,\ldots,n.
\]
If there exists at least one pair $i,j \in [n]$ such that $i \neq j$, $|i-j| \leqslant l$, $q_{i,n} \geqslant q_{j,n}$ and $\delta < q_{i,n}-q_{j,n}$ then
\[
\ep \geqslant \max_{\substack{1 \leqslant i,j \leqslant n \\ |i-j| \leqslant l}} \ln\left\{\frac{q_{i,n}-\delta}{q_{j,n}}\right\}.
\]
Otherwise $\ep \geq 0$. 
The inequalities are tight.
\end{theorem}
\begin{remark}
In this paper most of the time we find ourselves in the first situation. Usually $\delta$ will be already significantly smaller than $|q_{1,n}-q_{2,n}|$. Recall again that in practice one demands $\delta$ to be a very small constant (or even a value tending to $0$ with $n$).
\end{remark}
\begin{proof}
Since $M$ is $(\ep,\delta)$-DP w.r.t. $d_l$, the following inequality has to be satisfied for all $S \subseteq \C$ and for all $\rho \in \S_n$ such that $d_l(\sigma,\rho) = 1$ 
\begin{equation} \label{ineq:DP_basic}
\P[CM(\sigma, \tau) \in S] \leqslant e^{\ep} \P[CM(\rho, \tau) \in S] + \delta.
\end{equation}
Equivalently, we will work with the inequality
\begin{equation} \label{ineq:DP_eps}
e^{\ep} \geqslant \frac{\P[CM(\sigma,\tau) \in S] - \delta}{\P[CM(\rho,\tau) \in S]}
\end{equation}
assuming that $\P[CM(\rho,\tau) \in S] \neq 0$. (When $\P[CM(\rho,\tau) \in S] = 0$ the inequality~(\ref{ineq:DP_basic}) holds with any $\ep \geqslant 0$.) Thus let us investigate what is the maximal value that the right-hand side of the inequality (\ref{ineq:DP_eps}) may attain. Since $d_l(\sigma,\rho) = 1$, let us express $\rho$ as
\[
\rho = (\sigma_i~\sigma_j) \circ \sigma = (\sigma_1, \ldots, \sigma_{i-1}, \sigma_{j}, \sigma_{i+1}, \ldots, \sigma_i, \ldots, \sigma_n),\]
where $i,j \in [n]$, $i \neq j$ and $|i-j| \leqslant l$. For $k \in [n] \setminus \{i,j\}$ we have
\[
\P[CM(\rho,\tau)=C_{\sigma_k}] = \P[CM(\sigma,\tau)=C_{\sigma_k}] = q_{k,n}.
\]
In the remaining cases $\P[CM(\rho,\tau)=C_{\sigma_i}] = q_{j,n}$ and $\P[CM(\rho,\tau)=C_{\sigma_{j}}] = q_{i,n}$. Note that whenever $S$ neither contains $C_{\sigma_i}$ nor $C_{\sigma_j}$, the probabilities $\P[CM(\sigma,\tau) \in S]$ and $\P[CM(\rho,\tau) \in S]$ are equal and the above inequality holds with any $\ep \geqslant 0$. The situation is analogous whenever $S$ includes both, $C_{\sigma_i}$ and $C_{\sigma_j}$. Thus let us consider $S$ such that $C_{\sigma_i} \in S$ and $C_{\sigma_{j}} \notin S$ (the symmetric case with $C_{\sigma_j} \in S$ and $C_{\sigma_i} \notin S$ is analogous). We can write
\[
\P[CM(\sigma,\tau) \in S] = q_{i,n} + q \quad \textnormal{and} \quad \P[C M(\rho,\tau) \in S] = q_{j,n} + q,
\]
where $q = \P[CM(\sigma,\tau) \in S\setminus\{C_{\sigma_i}\}] = \P[CM(\rho,\tau) \in S\setminus\{C_{\sigma_i}\}]$. We have
\[
	\frac{\P[CM(\sigma,\tau) \in S] - \delta}{\P[CM(\rho,\tau) \in S]} = \frac{q_{i,n}+q-\delta}{q_{j,n}+q} =: f(q).
\]
Note that $f'(q) = \frac{q_{j,n}-q_{i,n}+\delta}{(q+q_{j,n})^2}$. Thus whenever $q_{i,n} \geqslant q_{j,n}$ and $\delta < q_{i,n}-q_{j,n}$ the function $f$ is decreasing. In the remaining cases $f$ is weakly increasing. Therefore if $q_{i,n} \geqslant q_{j,n}$ and $\delta < q_{i,n}-q_{j,n}$ we get
\[
\frac{\P[CM(\sigma,\tau) \in S] - \delta}{\P[CM(\rho,\tau) \in S]} \leqslant \frac{q_{i,n}-\delta}{q_{j,n}}
\]
which means that we maximize the expression setting $q=0$ (i.e., setting $S=\{C_{\sigma_i}\}$). The right-hand side of the above inequality is greater than or equal to 1. In the remaining cases (when $q_{i,n} \geqslant q_{j,n}$ and $\delta \geqslant q_{i,n}-q_{j,n}$ or when $q_{j,n} > q_{i,n}$) we get
\[
\frac{\P[CM(\sigma,\tau) \in S] - \delta}{\P[CM(\rho,\tau) \in S]} \leqslant \frac{1-q_{j,n}-\delta}{1-q_{i,n}}
\]
which means that we maximize the expression setting $q=1-q_{i,n}-q_{j,n}$ (i.e., setting $S = \C \setminus\{C_{\sigma_j}\}$). Note that the right-hand side of the above inequality is smaller than or equal to~1. Thus in this case the inequality (\ref{ineq:DP_eps}) holds for $\ep \geq 0$.

We conclude that whenever there exists at least one pair $i,j \in [n]$ such that $i \neq j$, $|i-j| \leqslant l$, $q_{i,n} \geqslant q_{j,n}$ and $\delta < q_{i,n}-q_{j,n}$ then
\[
\ep \geqslant \max_{\substack{1 \leqslant i,j \leqslant n \\ |i-j| \leqslant l}} \ln\left\{\frac{q_{i,n}-\delta}{q_{j,n}}\right\}.
\]
Otherwise $\ep \geq 0$. Note that in the proof in both cases we have indicated $S$ realizing the maximum. Therefore the bounds are tight.
\qed
\end{proof}

 \section{Hiding preferences}\label{sect:results}

Each optimal stopping algorithm is $(\ep,\delta)$-differentially private at some level, i.e., for some values of $\ep$ and $\delta$. These values will be often too high to meet the user's expectations. What the user can do is to resign from the optimality of the algorithm (however try to keep the accuracy of the algorithm at some acceptable level) gaining a higher level of privacy. It can be achieved by a careful modification of the distribution of the outcome of the algorithm. Analyzing the definition of differential privacy one can deduce that the closer this distribution to the uniform one is, the smaller the values of $\ep$ and $\delta$ in the definition of differential privacy may become. E.g., below in Fact \ref{fact:uniform_private} we explain that the algorithm with uniform outcome is $(0,0)$-differentially private regardless of the metric. Thus what the user should do in order to achieve the desired privacy level is to modify the algorithm $M^{opt}$ such that the distribution of its outcome comes in some sense closer to the uniform distribution. 

In this section we analyze a natural mechanism transforming an arbitrary optimal stopping time $M^{opt}$ into the algorithm meeting stricter privacy requirements, yet preserving some level of accuracy. It is equipped with a parameter $p \in [0,1]$ controlling the smooth transition between optimality and $(0,0)$-DP.


\begin{df} \label{df:pmix}
	Let $M^{''}$ and $M'$ be two online stopping algorithms for the same stopping problem. A $p$-mix on $M^{''}$ and $M'$ is defined as follows. We toss a coin that comes down heads with probability $p$. If it comes down heads, we play according to $M^{''}$. If it comes down tails, we play according to $M'$.
\end{df}

\begin{df} \label{df:blind}
We call the algorithm $\tilde{M}$ a blind choice if for a fixed $\sigma \in \S_n$ and for any $\tau \in \S_n$ it always stops at $\tau_1$, i.e., $C\tilde{M}(\sigma,\tau)=C_{\tau_1}$, equivalently $\tilde{M}(\sigma, \tau) = 1$.
\end{df}

\begin{fact} \label{fact:uniform_private}
A blind choice $\tilde{M}$ is $(0,0)$-differentially private regardless of the metric we use.
\end{fact}
\begin{proof}
Note that for a fixed $\sigma$ and any $C \in \C$ we have $\P[C\tilde{M}(\sigma,\tau)=C] = 1/n$. Indeed, the candidate $C$ will be selected by $\tilde{M}$ if and only if it is the first candidate in the time ordering $\tau$. Time ordering $\tau$ is chosen uniformly at random from $\T_n$, thus the probability that it starts with $C$ is $1/n$. In particular, for $S \subseteq \C$ we get $\P[C\tilde{M}(\sigma,\tau) \in S] = |S|/n$ (this probability is independent of $\sigma$). Thus for fixed $\sigma ,\rho \in \S_n$ and for any $S \subseteq \C$ we always get $\P[C\tilde{M}(\sigma,\tau) \in S] = \P[C\tilde{M}(\rho,\tau) \in S] = |S|/n$ thus
\[
\P[C\tilde{M}(\sigma,\tau) \in S] \leqslant  e^0 \P[C\tilde{M}(\rho,\tau) \in S] + 0.
\]
\qed
\end{proof}

Note that if we consider a $p$-mix on $M^{opt}$ and $\tilde{M}$ (i.e., a $p$-mix on an optimal algorithm and a blind choice) then we get a controllable by $p$ algorithm being a trade-off between two extremes. One of them is optimality (the case when $p=1$) and the other one is $(0,0)$-differential privacy (the case when $p=0$). Setting higher $p$ means relaxing the requirements for $(\ep,\delta)$-differential privacy but at the same time obtaining a larger probability of choosing the proper candidate. Setting smaller $p$ we resign from the high accuracy of the algorithm but we gain a higher level of privacy.

Some general results on the accuracy and the privacy of a $p$-mix algorithm are given in Section \ref{sec:general_pmix}. The detailed analysis of a classical case, i.e., of a $p$-mix on the optimal algorithm for the secretary problem and a blind choice, is given in Section \ref{sec:secretary_pmix}.

\subsection{Accuracy and differential privacy of a $p$-mix $M$} \label{sec:general_pmix}

In this section we formulate some general results on the accuracy and the privacy of a $p$-mix algorithm. We start with a simple fact about the minimum level to which the accuracy of this algorithm may drop.


\begin{fact} \label{fact:mix_efficiency}
Fix $\sigma \in \S_n$ and let $\tau$ be chosen uniformly at random from $\S_n$. Let $D \subseteq \C$ be the set of the desired candidates from Definition \ref{df:opt_stop}. Let $M^{opt}$ be the optimal stopping time and $M'$ any other stopping algorithm for this problem. Let $p \in [0,1]$ and $M$ be the $p$-mix on $M^{opt}$ and $M'$. Then
\[
\P[CM(\sigma,\tau) \in D] \geqslant p \cdot \P[CM^{opt}(\sigma,\tau) \in D].
\]
\end{fact}
\begin{proof}
Obviously, by the definition of a $p$-mix algorithm we get
\begin{align*}
\P[CM(\sigma,\tau) \in D] & = p \cdot \P[CM^{opt}(\sigma,\tau) \in D] + (1-p) \cdot \P[CM'(\sigma,\tau) \in D] \\
& \geqslant p \cdot \P[CM^{opt}(\sigma,\tau) \in D].
\end{align*}
\qed
\end{proof}


Assume that the optimal algorithm for some stopping problem is $(\ep,\delta)$ - differentially private. The following theorem explains how differential privacy improves (i.e., how parameters $\ep$ and $\delta$ drop) if we mix this optimal algorithm with a strategy whose outcome distribution is uniform.

\begin{theorem} \label{thm:optimal}
	Let $\varepsilon \geq 0$ and $\delta \in [0,1]$. Fix $\sigma \in \S_n$ and let $\tau$ be chosen uniformly at random from $\S_n$. Let also $M^{opt}$ be some optimal stopping algorithm. Consider a metric space $(\S_n, d_l)$. Assume that $M^{opt}$ is $(\ep,\delta)$-differentially private w.r.t. metric $d_l$. Let $M'$ be the algorithm whose outcome distribution is uniform, i.e., for all $k~\in~[n]$ $
	\P[CM'(\sigma,\tau) = C_{\sigma_{k}}] = 1/n$ (e.g., it may be a blind choice $\tilde{M}$).
	Then the $p$-mix $M$ on $M^{opt}$ and $M'$ is $\left(\ln \left(e^{\varepsilon}-\frac{(1-p)(e^{\varepsilon}-1)}{n}\right), p \cdot \delta\right)$-differentially private w.r.t. metric $d_l$.
\end{theorem}
\begin{proof}
	The algorithm $M^{opt}$ is $(\ep,\delta)$-differentially private w.r.t. metric $d_l$ thus for all $\rho \in \S_n$ such that $d_l(\sigma,\rho) \leqslant 1$ and for all $S \subseteq \C$ we have
	\[
	\P[CM^{opt}(\sigma,\tau) \in S] \leqslant e^{\ep} \P[CM^{opt}(\rho,\tau) \in S] + \delta.
	\]
	Moreover, by the definition of a $p$-mix for every $\pi \in \S_n$ we get
	\[
	\P[CM(\pi,\tau) \in S] = p \cdot \P[CM^{opt}(\pi,\tau) \in S] + (1-p) \cdot \P[CM'(\pi,\tau) \in S].
	\]
	Hence for all $\rho \in \S_n$ such that $d_l(\sigma,\rho) \leqslant 1$ and for all $S \subseteq \C$ such that $S \neq \emptyset$ (if $S = \emptyset$ then the differential privacy inequality for $M$ holds for any $\varepsilon \geq 0$ and any $\delta \geq 0$)
	\[ 
		\begin{split}
			\P[CM(\sigma,\tau) & \in S] = p \cdot \P[CM^{opt}(\sigma,\tau) \in S] + (1-p) \cdot \P[CM'(\sigma,\tau) \in S] \\
			& \leqslant p \cdot e^{\ep} \cdot \P[CM^{opt}(\rho,\tau) \in S] + p\cdot \delta + (1-p) \cdot \P[CM'(\sigma,\tau) \in S] \\
			& = e^{\ep} (\P[CM(\rho,\tau) \in S] - (1-p) \cdot \P[CM'(\rho,\tau) \in S]) \\
			& \quad + p\cdot \delta + (1-p) \cdot \P[CM'(\sigma,\tau) \in S] \\
			& = e^{\ep} \P[CM(\rho,\tau) \in S] + p\cdot \delta - (1-p) (e^{\ep}-1)\frac{|S|}{n} \\
			& \leqslant \left(e^{\ep} - \frac{(1-p)(e^{\ep}-1)}{n}\right) \P[CM(\rho,\tau) \in S] + p\cdot \delta. \quad \quad \quad \quad \quad \quad \qed
		\end{split} 
	\] 
\end{proof}

\subsection{Trade-off between optimality and differential privacy in the secretary problem} \label{sec:secretary_pmix}

This section is an analytical discussion about the optimal stopping algorithm for the classical secretary problem in the context of  differential privacy. Below we present a detailed analysis of the accuracy and differential privacy of a $p$-mix on $M^*$ and $\tilde{M}$, where $M^*$ is the optimal solution for the secretary problem, and $\tilde{M}$ is the blind choice. Recall that in the secretary problem one aims at choosing only the best out of all $n$ candidates (i.e., for a fixed $\sigma \in \S_n$ at selecting $\sigma_1$). Any other choice is interpreted as a loss. Let us start with a simple fact about the accuracy of the $p$-mix on $M^*$ and $\tilde{M}$.

\begin{fact} \label{fact:efficiency_pmix}
Fix $\sigma \in \S_n$. Let $M$ be a $p$-mix on $M^*$ and $\tilde{M}$, where $M^*$ is the optimal algorithm for the secretary problem, and $\tilde{M}$ is the blind choice. Then
\[
\P[CM(\sigma,\tau) = C_{\sigma_1}] = p \cdot \P[CM^*(\sigma,\tau) = C_{\sigma_1}] + \frac{1-p}{n} \sim \frac{p}{e}.
\]
\end{fact}
\begin{proof}
The result follows straight from the definition of a $p$-mix and Fact \ref{fact:kth_cand}.
\qed
\end{proof}

Before we move on to analyzing differential privacy of the classical $p$-mix on $M^*$ and $\tilde{M}$, let us introduce some simplifications in notation.
Throughout this section $p \in [0,1]$ is always a constant and $M$ is always a $p$-mix on $M^*$ and $\tilde{M}$. (The case $p=0$ when $M$ is just a blind choice was already discussed thus we will often assume $p \in (0,1]$.) We also introduce a shorter notation for the probabilities that $M^*$ or $M$ select the $k\textsuperscript{th}$ best candidate, namely, for $\sigma \in \S_n$ and $\tau$ chosen uniformly at random from $\S_n$
\[
r_{k,n} = \P[CM^*(\sigma,\tau) = C_{\sigma_k}] \quad \textnormal{and} \quad q_{k,n} = \P[CM(\sigma,\tau) = C_{\sigma_k}].
\]
By this notation, following Definition \ref{df:pmix}, we can write
\begin{equation} \label{lemma:pmix_prob}
q_{k,n} = p \cdot r_{k,n} + \frac{1-p}{n}.
\end{equation}
Therefore, by Theorem \ref{thm:kconst_limit} (see the Appendix), we formulate the following corollary.
\begin{corollary} \label{cor:pmix_asymp}
Let $\sigma \in \S_n$ and $\tau$ be chosen uniformly at random from $\S_n$. Let $p \in [0,1]$ be a constant and let $M$ be a $p$-mix on $M^*$ and $\tilde{M}$. Then for 
$k = k(n)=o(n)$
\[
q_{k,n}  \sim \frac{p}{e} \cdot \sum_{s=k}^{\infty} \frac{1}{s} \left(1-\frac{1}{e}\right)^s, \quad \textit{in particular} \quad q_{1,n}  \sim \frac{p}{e}.
\]
\end{corollary}
For transparency of notation let also
\[
a_k = \sum_{s=k}^{\infty}\frac{1}{s} \left(1-\frac{1}{e}\right)^s.
\]
Note that the above sum always converges, in particular $a_1=1$ and $a_2 = 1/e$. By Fact~\ref{fact:kth_cand} and Corollary \ref{cor:pmix_asymp} for 
$k = k(n)=o(n)$ we can simply write 
\[
r_{k,n} \sim \frac{1}{e} \cdot a_k \quad \textnormal{and} \quad q_{k,n} \sim \frac{p}{e} \cdot a_k.
\]

In Theorem \ref{thm:main_eps} and the two following corollaries we give constraints for $\ep, \delta$ and $p$ which guarantee that a $p$-mix $M$ is $(\ep,\delta)$-differentially private with respect to metric $d_l$. 
Recall that we always assume that $l$ is a constant (consult Theorem \ref{thm:l_inf}). The following two technical lemmas will be helpful when proving the main theorem.

\begin{lemma} \label{lemma:seq_dec}
Fix $\sigma \in \S_n$ and let $\tau$ be chosen uniformly at random from $\S_n$. 
The sequence $\{q_{k,n}\}_{k \in [n]}$ is non-increasing in $k$.
\end{lemma}
\begin{proof}
By Corollary \ref{cor:rkn_dec} (see the Appendix) we know that the sequence $\{r_{k,n}\}_{k \in [n]}$ is non-increasing in $k$. By (\ref{lemma:pmix_prob}) we have
\[
q_{k,n} - q_{k+1,n} = p \cdot (r_{k,n} - r_{k+1,n}) \geqslant 0.
\]
\qed
\end{proof}

\begin{lemma} \label{lemma:max_ratio}
Fix $\sigma \in \S_n$ and let $\tau$ be chosen uniformly at random from $\S_n$. Let $l\geqslant 1$ be a constant, and let $\delta \in [0,1]$. For any $k \in [n-1]$, if $n \geqslant 7$ then
\[
\frac{q_{1,n}-\delta}{q_{l+1,n}} \geqslant \frac{q_{k,n}-\delta}{q_{k+l,n}}.
\]
\end{lemma}
\begin{proof}
First, we are going to show that for $k \geqslant 2$
\begin{equation} \label{ineq:qlimits}
\lim_{n \to \infty} \frac{q_{1,n}}{q_{l+1,n}} > \lim_{n \to \infty} \frac{q_{k,n}}{q_{k+l,n}}.
\end{equation}
By (\ref{lemma:pmix_prob}) and by Fact \ref{fact:kth_cand} we have
\[
\lim_{n \to \infty} \frac{q_{1,n}}{q_{l+1,n}} = \lim_{n \to \infty} \frac{r_{1,n}}{r_{l+1,n}} = \frac{\frac{1}{e} a_{1}}{\frac{1}{e} a_{\ell+1}} = \frac{1}{a_{l+1}} 
\]
and
\[ \lim_{n \to \infty} \frac{q_{k,n}}{q_{k+l,n}} = \lim_{n \to \infty} \frac{p \cdot r_{k,n} + \frac{1-p}{n}}{p \cdot r_{k+l,n} + \frac{1-p}{n}}.
\]
By Theorem \ref{thm:ratio_rnk} (see the Appendix) we know that for $k \geqslant 2$ we have $\lim_{n \to \infty} \frac{r_{k,n}}{r_{k+l,n}}  < \frac{1}{a_{l+1}}$. Thus for sufficiently large $n$ (one can verify that $n \geqslant 7$ is enough) we can write $r_{k,n} < \frac{r_{k+l,n}}{a_{l+1}}$ and thereby get (note that $a_{l+1} < 1$)
\begin{align*}
\frac{p \cdot r_{k,n} + \frac{1-p}{n}}{p \cdot r_{k+l,n} + \frac{1-p}{n}} & < \frac{p \cdot \frac {r_{k+l,n}}{a_{l+1}} + \frac{1-p}{n}}{p \cdot r_{k+l,n} + \frac{1-p}{n}} \\
& \xrightarrow{n \to \infty}
\begin{cases}
 \frac{p \cdot \frac{c}{a_{l+1}} + 1 - p}{p \cdot c + 1 - p} < \frac{1}{a_{l+1}} \quad & \textnormal{for} \quad r_{k+l,n} = c \cdot 1/n + o(1/n), \\
 1 < \frac{1}{a_{l+1}} \quad & \textnormal{for} \quad r_{k+l,n} = o(1/n).
\end{cases}
\end{align*}
Whenever $r_{k+l,n} = \omega(1/n)$ we also have $r_{k,n} = \omega(1/n)$ (indeed, by Corollary \ref{cor:rkn_dec} the sequence $r_{k,n}$ is non-increasing in $k$) and again by Theorem \ref{thm:ratio_rnk} (see the Appendix) we get
\[
\lim_{n \to \infty} \frac{p \cdot r_{k,n} + \frac{1-p}{n}}{p \cdot r_{k+l,n} + \frac{1-p}{n}} = \lim_{n \to \infty} \frac{r_{k,n}}{r_{k+l,n}}  < \frac{1}{a_{l+1}}.
\]

We conclude that for any $k \in [n-1]$ and sufficiently large $n$ (again $n \geqslant 7$ is enough) we get $\frac{q_{1,n}}{q_{l+1,n}} \geqslant \frac{q_{k,n}}{q_{k+l,n}}$ and, by Lemma \ref{lemma:seq_dec},
\[
\frac{q_{1,n}-\delta}{q_{l+1,n}} = \frac{q_{1,n}}{q_{l+1,n}} - \frac{\delta}{q_{l+1,n}} \geqslant \frac{q_{1,n}}{q_{l+1,n}} - \frac{\delta}{q_{k+l,n}} \geqslant \frac{q_{k,n}}{q_{k+l,n}} - \frac{\delta}{q_{k+l,n}} = \frac{q_{k,n}-\delta}{q_{k+l,n}}.
\]
\qed
\end{proof}

\begin{theorem} \label{thm:main_eps}
Fix $\sigma \in \S_n$ and let $\tau$ be chosen uniformly at random from $\S_n$. Consider a metric space $(\S_n,d_l)$ for $l \geqslant 1$ being a constant. Let $p \in (0,1]$ and $\delta \in [0,1]$. If $n \geqslant 7$ and 
\[
\ep \geqslant \begin{cases}
\ln\left( \frac{q_{1,n} - \delta}{q_{l+1,n}} \right) \sim  \ln\left( \frac{p - \delta \cdot e}{a_{l+1} \cdot p} \right) & \quad \textnormal{for} \quad \delta < q_{1,n}-q_{l+1,n} \\
0 & \quad \textnormal{for} \quad \delta \geqslant q_{1,n}-q_{l+1,n},
\end{cases}
\]
then the $p$-mix $M$ (on $M^*$ and $\tilde{M}$) is $(\ep,\delta)$-differentially private. The bounds are tight.
\end{theorem}
\begin{proof}
	By Theorem \ref{thm:gen_eps} we know that if there exists at least one pair $i,j \in [n]$ such that $i \neq j$, $|i-j| \leqslant l$, $q_{i,n} \geqslant q_{j,n}$ and $\delta < q_{i,n}-q_{j,n}$ then
	\[
	\ep \geqslant \max_{\substack{1 \leqslant i,j \leqslant n \\ |i-j| \leqslant l}} \ln\left\{\frac{q_{i,n}-\delta}{q_{j,n}}\right\}.
	\] 
	Otherwise $\ep \geq 0$. Consequently, by Lemma \ref{lemma:seq_dec} and Lemma \ref{lemma:max_ratio} we get
	\[
	\ep \geqslant \begin{cases}
		\ln\left( \frac{q_{1,n} - \delta}{q_{l+1,n}} \right) & \quad \textnormal{for} \quad \delta < q_{1,n}-q_{l+1,n} \\
		0 & \quad \textnormal{for} \quad \delta \geqslant q_{1,n}-q_{l+1,n}.
	\end{cases}
	\]
	Additionally, by Corollary \ref{cor:pmix_asymp} we get
	\[
	\ln \left( \frac{q_{1,n} - \delta}{q_{l+1,n}} \right) \sim \ln \left( \frac{p - \delta \cdot e}{a_{l+1} \cdot p} \right).
	\]
\qed
\end{proof}

Figure \ref{fig:p_single} shows the shape of the asymptotic region of pairs $(\ep,\delta)$ for which the $p$-mix $M$ is $(\ep,\delta)$-differentially private.
\begin{figure}[!ht]
\centering
\includegraphics[width=0.5\textwidth]{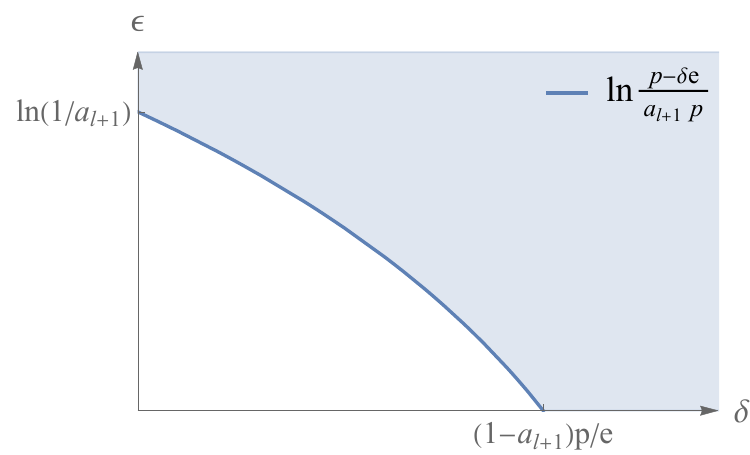}
\caption{The shaded area is an asymptotic region in which the $p$-mix $M$ is $(\ep,\delta)$-differentially private.}
\label{fig:p_single}
\end{figure}

\begin{figure}[ht]
	\begin{minipage}[b]{0.49\textwidth}
		\includegraphics[width=0.95\textwidth]{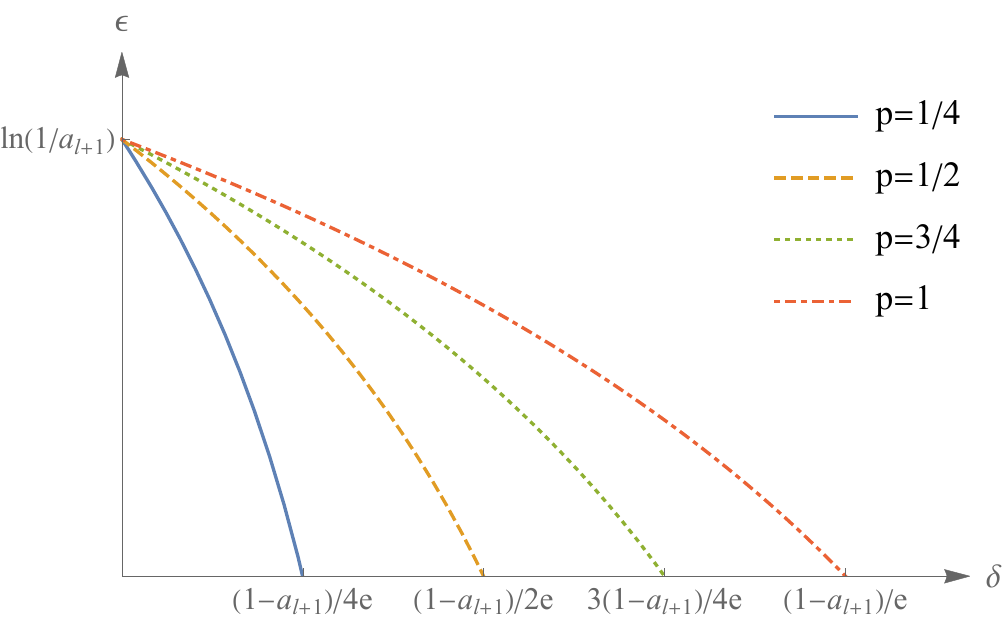}
		\caption{The boundaries of the asymptotic regions in which the $p$-mix $M$ is $(\ep,\delta)$-differentially private for a given $l$ and various values of $p$.}
		\label{fig:p_comparison}
	\end{minipage}
	\hspace{10pt}
	\begin{minipage}[b]{0.49\textwidth}
		\includegraphics[width=0.95\textwidth]{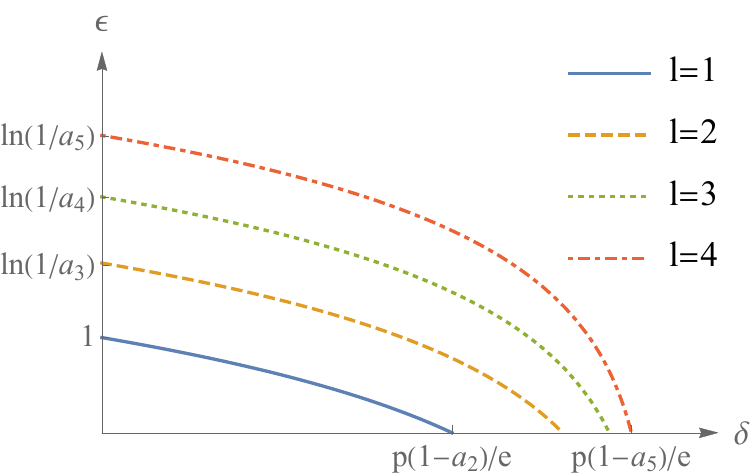}
		\caption{The boundaries of the asymptotic regions in which the $p$-mix $M$ is $(\ep,\delta)$-differentially private for a given $p$ and various values of $l$.}
		\label{fig:l_comparison}
	\end{minipage}
\end{figure}

Figure \ref{fig:p_comparison} shows how the boundaries of the asymptotic regions in which the $p$-mix $M$ is $(\ep,\delta)$-differentially private change with $p$ for a given~$l$.
Figure \ref{fig:l_comparison} shows how the boundaries of an asymptotic region in which the $p$-mix $M$ is $(\ep,\delta)$-differentially private change with $l$ for a  given $p$. 

\begin{corollary} \label{cor:delta2}
Fix $\sigma \in \S_n$ and let $\tau$ be chosen uniformly at random from $\S_n$. Consider a metric space $(\S_n,d_l)$ for $l \geqslant 1$ being a constant. Let $p\in (0,1]$ and $\ep \geqslant 0$. If $n \geqslant 7$ and 
\[
\delta \geqslant \begin{cases}
    q_{1,n} - e^{\ep}q_{l+1,n} \sim \frac{p}{e}\left(1-a_{l+1}e^{\ep}\right) & \quad \textnormal{for} \quad \ep < \ln\left(\frac{q_{1,n}}{q_{l+1,n}}\right)\\
		0 & \quad \textnormal{for} \quad \ep \geqslant \ln\left(\frac{q_{1,n}}{q_{l+1,n}}\right),
		\end{cases}
\]
then the $p$-mix $M$ (on $M^*$ and $\tilde{M}$) is $(\ep,\delta)$-differentially private.
\end{corollary}
\begin{proof}
Note that the inequality $\delta \geqslant q_{1,n} - e^{\ep}q_{l+1,n}$ is equivalent to the inequality $\ep \geqslant \ln \left( \frac{q_{1,n} - \delta}{q_{l+1,n}} \right)$ from Theorem \ref{thm:main_eps} and by Corollary \ref{cor:pmix_asymp} we get $q_{1,n} - e^{\ep}q_{l+1,n} \sim \frac{p}{e}\left(1-a_{l+1}e^{\ep}\right)$.
\qed
\end{proof}

\begin{corollary}
Fix $\sigma \in \S_n$ and let $\tau$ be chosen uniformly at random from $\S_n$. Consider a metric space $(\S_n,d_l)$ for $l \geqslant 1$ being a constant. Let $\delta \in [0,1]$ and $\ep \geqslant 0$. If $n \geqslant 7$ and 
\[
p \leqslant \begin{cases}
   \frac{\delta + \frac{1}{n}(e^{\ep}-1)}{r_{1,n}-e^{\ep}r_{l+1,n} + \frac{1}{n}(e^{\ep}-1)} \sim \frac{e \cdot  \delta}{1-e^{\ep} a_{l+1}} & \quad \textnormal{for} \quad \delta < r_{1,n}-e^{\ep}r_{l+1,n}\\
		1 & \quad \textnormal{for} \quad \delta \geqslant r_{1,n}-e^{\ep}r_{l+1,n},
		\end{cases}
\]
then  the $p$-mix $M$ (on $M^*$ and $\tilde{M}$) is $(\ep,\delta)$-differentially private.
\end{corollary}
\begin{proof}
By Corollary \ref{cor:delta2} we know that for $\ep \geqslant 0$ the $p$-mix $M$ is $(\ep,\delta)$-differentially private if only $\delta \geqslant q_{1,n} - e^{\ep}q_{l+1,n}$. By (\ref{lemma:pmix_prob}) we can rewrite it as
\[
\delta \geqslant p \cdot r_{1,n} + \frac{1-p}{n} - e^{\ep}\left(p \cdot r_{l+1,n} + \frac{1-p}{n}\right) 
\]
which is equivalent to
\[
p \leqslant \frac{\delta + \frac{1}{n}(e^{\ep}-1)}{r_{1,n}-e^{\ep}r_{l+1,n} + \frac{1}{n}(e^{\ep}-1)} \sim \frac{e \cdot  \delta}{1-e^{\ep} a_{l+1}},
\]
where asymptotics follows from Fact \ref{fact:kth_cand}. \qed
\end{proof}

\begin{figure}[ht]
	\begin{minipage}[b]{0.45\textwidth}
		\includegraphics[width=0.95\textwidth]{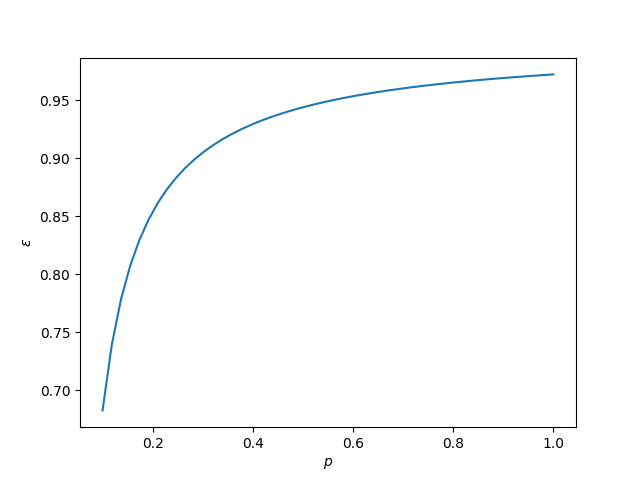}
		\caption{Privacy parameter $\varepsilon$ of the $p$-mix $M$ for $\delta=0.01$.}
		\label{fig:thm2}
	\end{minipage}
	\hspace{5pt}
	\begin{minipage}[b]{0.45\textwidth}
		\includegraphics[width=0.95\textwidth]{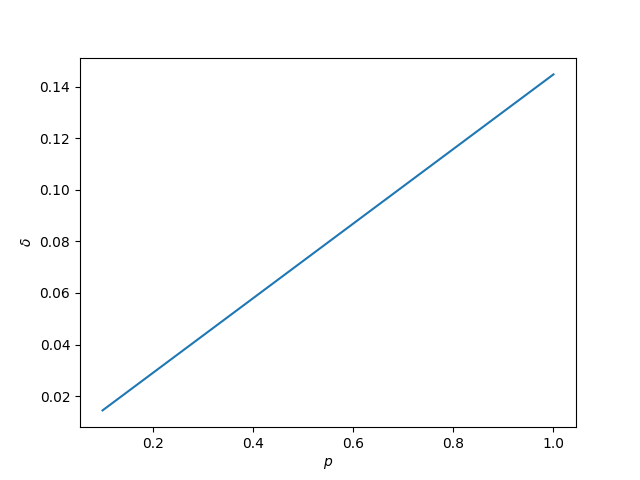}
		\caption{Privacy parameter $\delta$ of the $p$-mix $M$ for $\varepsilon=0.5$.}
		\label{fig:cor2}
	\end{minipage}
	\vskip\baselineskip
	\vspace{-15pt}
	\centering
	\begin{minipage}[b]{0.45\textwidth}
		\includegraphics[width=0.95\textwidth]{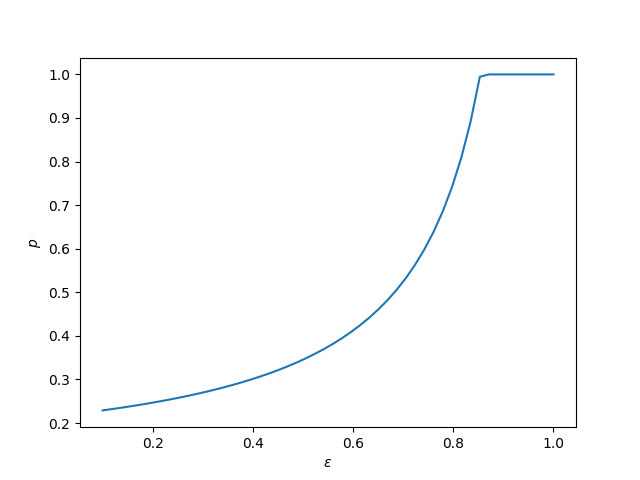}
		\caption{Parameter $p$ of the $p$-mix $M$ for $\delta~=~0.05$.}
		\label{fig:cor3}
	\end{minipage}   
	\vspace{-22pt}
\end{figure}

Below we present also a few explicit examples of the relations between the parameters $\ep, \delta$ and $p$. All the examples consider asymptotic results w.r.t. metric $d_1$. 
When needed the parameters were set to $\varepsilon = 0.5$, and $\delta = 0.01$ or $\delta = 0.05$ commonly used in the analysis of differential privacy offered by a wide range of mechanisms (see e.g. \cite{Zhu,DworkAlgo}). Note however that the calibration of $\varepsilon,\delta$ parameters depends strongly on the considered scenarios. In a one-shot response mechanism such values may be adequate, however, if a given mechanism is assumed to be used many times or combined with other sources of knowledge by the adversary we often need to consider much lower values (\cite{Haeberlen}).
Figure~\ref{fig:thm2} shows the parameter $\ep$ as a function of $p$ for which the $p$-mix $M$ is $(\ep,0.01)$-differentially private. Note that here $\varepsilon$ does not exceed $1$ even for $p=1$, thus when the selector plays simply the optimal algorithm.  Figure~\ref{fig:cor2} shows the parameter $\delta$ as a function of $p$ for which the~$p$-mix~$M$ is $(0.5,\delta)$-differentially private. Note that if we demand $\delta$ to be at most $0.05$, then $p$ has to be around $0.35$.
Figure \ref{fig:cor3} shows how small the parameter $p$ has to be for the $p$-mix $M$ to obtain $(\varepsilon,0.05)$-differential privacy. Note that if we accept $\varepsilon \approx 0.85$ or higher, then already $p=1$ (thus simply the optimal strategy) is sufficient.





 \section{Conclusions and Future Work}\label{sect:conclusion}

\vspace{-3pt}

In this paper we have investigated the optimal stopping algorithms from the information hiding perspective. We have proposed a natural mechanism constructing suboptimal stopping algorithms to give better privacy properties. We have analyzed its effectiveness and privacy, and applied it to the classical secretary problem. This work can be seen as the first step towards the differentially private stopping algorithms. 



The problem we consider might, at first glance, resemble the differentially private auction problem (see~\cite{mcsherry2007mechanism}). Note however, that here we do not have the scores, the preferences are given by the permutation and they are only comparable, not quantifiable, so the exponential mechanism cannot be used. Moreover, due to the online nature of the problem, we cannot have full knowledge during the whole procedure. We also have to define a metric of similarity for inputs.

Here we have concentrated on the classical secretary problem. Nevertheless, the optimal stopping literature offers the whole variety of different models. One could, e.g., analyze the algorithm that optimizes the expected rank of the candidate (intuitively, we do not necessarily require the best candidate but at least 'good enough'). It may have better information-hiding properties inherently. One could consider Gusein-Zade models with parameter $k$ in which we win when the selected candidate is in top $k$. 
It would be challenging to work on models with different ordering, e.g., using partially ordered sets instead of linear order as a qualification ordering. However, in this case it is unclear whether a reasonable metric for the similarity of preferences could be proposed.

Our approach was to keep the final algorithm simple, so we based it on known, optimal one. Elseways, one might propose an entirely different algorithm, that is not optimal but has better information-hiding properties or is more effective for given privacy parameters. Note also that essentially we have focused on hiding information about preferences, not the participation of a specific candidate. Another area of research could be investigating whether participation hiding is feasible in such circumstances. Even if not for all candidates, maybe it could be possible for the majority of them.

We believe that this paper opens an interesting new research area lying at the crossroads of online algorithms and differential privacy.

\paragraph{Acknowledgements}

This research was partially supported by Polish National Science Centre Grant 2018/29/B/ST6/02969.

%
%
%

\vspace{-7pt}

\bibliographystyle{splncs-url}
\bibliography{main}

\begin{thebibliography}{10}

\bibitem{Lindley}
Lindley, D.:
\newblock Dynamic programming and decision theory.
\newblock Appl. Stat. - J. Roy. St. C \textbf{10}(1) (1961)  39--51

\bibitem{Ferguson}
Ferguson, T.S.:
\newblock Who solved the secretary problem?
\newblock Statist. Sci. \textbf{4}(3) (1989)  282--289

\bibitem{Stadje}
Stadje, W.:
\newblock Efficient stopping of a random series of partially ordered points.
\newblock Multiple Criteria Decision Making Theory and Application. Lecture
  Notes in Economics and Mathematical Systems \textbf{177} (1980)  430--447

\bibitem{Gnedin}
Gnedin, A.V.:
\newblock Multicriteria extensions of the best choice problem: Sequential
  selec- tion without linear order.
\newblock Contemp. Math. \textbf{125} (1992)  153--172

\bibitem{bin_tree}
Morayne, M.:
\newblock Partial-order analogue of the secretary problem the binary tree case.
\newblock Discret. Math. \textbf{184}(1-3) (1998)  165--181

\bibitem{Garrod}
Garrod, B., Morris, R.:
\newblock The secretary problem on an unknown poset.
\newblock Random Struct. Algor. \textbf{43}(4) (2013)  429--451

\bibitem{univ_poset}
Preater, J.:
\newblock The best-choice problem for partially ordered objects.
\newblock Oper. Res. Lett. \textbf{25}(4) (1999)  187--190

\bibitem{Kozik}
Kozik, J.:
\newblock Dynamic threshold strategy for universal best choice problem.
\newblock Proceedings of 21st International Meeting on Probabilistic,
  Combinatorial, and Asymptotic Methods in the Analysis of Algorithms (2010)
  439--452

\bibitem{Freij_Wastlund}
Freij, R., W{\"a}stlund, J.:
\newblock Partially ordered secretaries.
\newblock Electron. Commun. Prob. \textbf{15} (2010)  504--507

\bibitem{Kubicki_Morayne_2005}
Kubicki, G., Morayne, M.:
\newblock Graph-theoretic generalization of the secretary problem: The directed
  path case.
\newblock SIAM Journal on Discrete Mathematics \textbf{19}(3) (2005)  622--632

\bibitem{Przykucki_Sulkowska_2010}
Przykucki, M., Sulkowska, M.:
\newblock Gusein-zade problem for directed path.
\newblock Discrete Optimization \textbf{7}(1) (2010)  13--20

\bibitem{Sulkowska_universal_2012}
Sulkowska, M.:
\newblock The best choice problem for upward directed graphs.
\newblock Discrete Optimization \textbf{9}(3) (2012)  200--204

\bibitem{Goddard_Kubiccy_2013}
Goddard, W., Kubicka, E., Kubicki, G.:
\newblock An efficient algorithm for stopping on a sink in a directed graph.
\newblock Operations Research Letters \textbf{41}(3) (2013)  238--240

\bibitem{Grzesik_Morayne_Sulkowska_2015}
Grzesik, A., Morayne, M., Sulkowska, M.:
\newblock From directed path to linear order---the best choice problem for
  powers of directed path.
\newblock SIAM Journal on Discrete Mathematics \textbf{29}(1) (2015)  500--513

\bibitem{Benevides_Sulkowska_2017}
Benevides, F.S., Sulkowska, M.:
\newblock Percolation and best-choice problem for powers of paths.
\newblock Journal of Applied Probability \textbf{54}(2) (2017)  343–362

\bibitem{matroid_secretary}
Babaioff, M., Immorlica, N., Kleinberg, R.:
\newblock Matroids, secretary problems, and online mechanisms.
\newblock In Bansal, N., Pruhs, K., Stein, C., eds.: Proceedings of the
  Eighteenth Annual {ACM-SIAM} Symposium on Discrete Algorithms, {SODA} 2007,
  New Orleans, Louisiana, USA, January 7-9, 2007, {SIAM} (2007)  434--443

\bibitem{matroid_STOC2016}
Rubinstein, A.:
\newblock Beyond matroids: secretary problem and prophet inequality with
  general constraints.
\newblock STOC '16, New York, NY, USA, Association for Computing Machinery
  (2016)  324–332

\bibitem{matroid_2018}
Babaioff, M., Immorlica, N., Kempe, D., Kleinberg, R.:
\newblock Matroid secretary problems.
\newblock Journal of the ACM \textbf{65}(6) (2018)

\bibitem{Chow}
Chow~Y.S., Moriguti~S., R.H., S.M., S.:
\newblock Optimal selection based on relative rank (the ``secretary problem'').
\newblock Isr. J. Math. \textbf{2} (1964)  81--90

\bibitem{SODA20}
Kaplan, H., Naori, D., Raz, D.:
\newblock Competitive analysis with a sample and the secretary problem.
\newblock In Chawla, S., ed.: Proceedings of the 2020 {ACM-SIAM} Symposium on
  Discrete Algorithms, SODA 2020, {SIAM} (2020)  2082--2095

\bibitem{SODA21}
Correa, J., Cristi, A., Feuilloley, L., Oosterwijk, T., Tsigonias-Dimitriadis,
  A.:
\newblock The secretary problem with independent sampling.
\newblock In: Proceedings of the 2021 ACM-SIAM Symposium on Discrete
  Algorithms, SODA 2021, {SIAM} (2021)  2047--2058

\bibitem{multi_secr_2022}
Besbes, O., Kanoria, Y., Kumar, A.:
\newblock The multi-secretary problem with many types.
\newblock In: Proceedings of the 23rd ACM Conference on Economics and
  Computation. EC '22, New York, NY, USA, Association for Computing Machinery
  (2022)  1146–1147

\bibitem{Albers_TCS_2021}
Albers, S., Ladewig, L.:
\newblock New results for the k-secretary problem.
\newblock Theoretical Computer Science \textbf{863} (2021)  102--119

\bibitem{Fujii_2023}
Fujii, K., Yoshida, Y.:
\newblock The secretary problem with predictions.
\newblock Mathematics of Operations Research \textbf{49}(2) (2023)  1241--1262

\bibitem{secretary_management_2024}
Correa, J., Cristi, A., Feuilloley, L., Oosterwijk, T., Tsigonias-Dimitriadis,
  A.:
\newblock The secretary problem with independent sampling.
\newblock Management Science \textbf{71}(4) (2024)  2778--2801

\bibitem{dwork2006calibrating}
Dwork, C., McSherry, F., Nissim, K., Smith, A.:
\newblock Calibrating noise to sensitivity in private data analysis.
\newblock In: TCC. Volume 3876., Springer (2006)  265--284

\bibitem{Dwork06}
Dwork, C.:
\newblock Differential privacy.
\newblock In: Automata, Languages and Programming, 33rd International
  Colloquium, {ICALP} 2006. (2006)  1--12

\bibitem{DworkAlgo}
Dwork, C., Roth, A.:
\newblock The algorithmic foundations of differential privacy.
\newblock Foundations and Trends in Theoretical Computer Science
  \textbf{9}(3-4) (2014)  pp. 211--407

\bibitem{dwork2006our}
Dwork, C., Kenthapadi, K., McSherry, F., Mironov, I., Naor, M.:
\newblock Our data, ourselves: Privacy via distributed noise generation.
\newblock In: Annual International Conference on the Theory and Applications of
  Cryptographic Techniques, Springer (2006)  486--503

\bibitem{dwork2009differential}
Dwork, C., Lei, J.:
\newblock Differential privacy and robust statistics.
\newblock In: STOC. Volume~9. (2009)  371--380

\bibitem{dwork2010differential}
Dwork, C., Naor, M., Pitassi, T., Rothblum, G.N.:
\newblock Differential privacy under continual observation.
\newblock In: Proceedings of the forty-second ACM symposium on Theory of
  computing, ACM (2010)  715--724

\bibitem{narayanan2009anonymizing}
Narayanan, A., Shmatikov, V.:
\newblock De-anonymizing social networks.
\newblock In: Security and Privacy, 2009 30th IEEE Symposium on, IEEE (2009)
  173--187

\bibitem{narayanan2010myths}
Narayanan, A., Shmatikov, V.:
\newblock Myths and fallacies of personally identifiable information.
\newblock Commun ACM \textbf{53}(6) (2010)  24--26

\bibitem{Rogerson}
Rogerson, P.:
\newblock Probabilities of choosing applicants of arbitrary rank in the
  secretary problem.
\newblock J. Appl. Probab. \textbf{24}(2) (1987)  527--533

\bibitem{Zhu}
Zhu, T., Li, G., Zhou, W., Yu, P.S.:
\newblock Differential Privacy and Applications. Volume~69 of Advances in
  Information Security.
\newblock Springer (2017)

\bibitem{Haeberlen}
Haeberlen, A., Pierce, B.C., Narayan, A.:
\newblock Differential privacy under fire.
\newblock In: 20th {USENIX} Security Symposium, San Francisco, CA, USA, August
  8-12, 2011, Proceedings, {USENIX} Association (2011)

\bibitem{mcsherry2007mechanism}
McSherry, F., Talwar, K.:
\newblock Mechanism design via differential privacy.
\newblock In: 48th Annual IEEE Symposium on Foundations of Computer Science
  (FOCS'07), IEEE (2007)  94--103

\end{thebibliography}
 
\newpage
\section*{Appendix}\label{sect:App}
\renewcommand{\thesubsection}{\Alph{subsection}}
The Appendix contains the technical results on the probabilities that the optimal stopping algorithm for the secretary problem selects the $k\textsuperscript{th}$ best candidate. The proofs of Theorems~\ref{thm:kth_cand}~and~\ref{thm:kconst_limit} (however in slightly different formulations) can be found in \cite{Rogerson}. 
\begin{theorem} \label{thm:kth_cand}
Let $M^*$ be the optimal stopping algorithm for the secretary problem. Fix $\sigma \in \S_n$ and let $\tau$ be chosen uniformly at random from $\S_n$. Let $k \in [n]$. The probability that $M^*$ selects the $k\textsuperscript{th}$ best candidate is given by
\[
r_{k,n} = \P[CM^*(\sigma, \tau) = C_{\sigma_k}] = \frac{t_n-1}{n} \left( \sum_{i=t_n}^{n-k+1}\frac{\binom{n-k}{i-1}}{\binom{n-1}{i-1}}\frac{1}{i-1} + \frac{1}{n-1} \right),
\]
where $t_n$ is the threshold from Fact \ref{fact:threshold}.
\end{theorem}

\begin{lemma} \label{lemma:rkn_difference}
Let $r_{k,n}$ be defined as in Theorem \ref{thm:kth_cand}. Then for $k \in [n-1]$
\[
r_{k+1,n} = r_{k,n} - \frac{t_n-1}{n} \cdot \frac{1}{k} \cdot \frac{\binom{n-t_n+1}{k}}{\binom{n-1}{k}}.
\]
\end{lemma}
\begin{proof}
Since $\frac{\binom{n-k-1}{i-2}}{\binom{n-2}{i-2}} = \frac{\binom{n-i}{k-1}}{\binom{n-2}{k-1}}$ we have
\begin{align*}
r_{k,n} - r_{k+1,n} & = \frac{t_n-1}{n} \sum_{i=t_n}^{n-k+1}\frac{\binom{n-k}{i-1}-\binom{n-k-1}{i-1}}{\binom{n-1}{i-1}}\frac{1}{i-1} \\
& = \frac{t_n-1}{n} \sum_{i=t_n}^{n-k+1}\frac{\binom{n-k-1}{i-2}}{\binom{n-1}{i-1}}\frac{1}{i-1} = \frac{t_n-1}{n} \sum_{i=t_n}^{n-k+1}\frac{\binom{n-k-1}{i-2}}{\binom{n-2}{i-2}}\frac{i-1}{n-1}\frac{1}{i-1} \\
& = \frac{t_n-1}{n} \frac{1}{n-1} \sum_{i=t_n}^{n-k+1}\frac{\binom{n-i}{k-1}}{\binom{n-2}{k-1}} = \frac{t_n-1}{n} \frac{1}{n-1} \frac{\binom{n-t_n+1}{k}}{\binom{n-2}{k-1}} = \\
& = \frac{t_n-1}{n} \cdot \frac{1}{k} \cdot \frac{\binom{n-t_n+1}{k}}{\binom{n-1}{k}}.
\end{align*}
\qed
\end{proof}

\begin{corollary} \label{cor:rkn_dec}
The sequence $r_{k,n}$ is non-increasing in $k$.
\end{corollary}
\begin{proof}
By Lemma \ref{lemma:rkn_difference} we get
\[
r_{k,n} - r_{k+1,n} = \frac{t_n-1}{n} \cdot \frac{1}{k} \frac{\binom{n-t_n+1}{k}}{\binom{n-1}{k}} \geqslant 0.
\]
\qed
\end{proof}


\begin{theorem} \label{thm:kconst_limit}
Let 
 $k = k(n) = o(n)$. Let $r_{k,n}$ be defined as in Theorem \ref{thm:kth_cand}. Then
\[
r_{k,n} \sim \frac{1}{e} \left(1 - \sum_{s=1}^{k-1} \frac{1}{s}\left(1-\frac{1}{e}\right)^s \right) = \frac{1}{e} \sum_{s=k}^{\infty} \frac{1}{s}\left(1-\frac{1}{e}\right)^s.
\]
\end{theorem}

\begin{lemma} \label{lemma:k_linear}
Let $r_{k,n}$ be defined as in Theorem \ref{thm:kth_cand}. Let $k=k(n) \leqslant n$ be a function linear in $n$. Then
\[
r_{k,n} \sim \frac{1}{e} \cdot \frac{1}{n}. 
\]
\end{lemma}
\begin{proof}
Recall that
\[
r_{k,n} = \frac{t_n-1}{n} \sum_{i=t_n}^{n-k+1}\frac{\binom{n-k}{i-1}}{\binom{n-1}{i-1}}\frac{1}{i-1} + \frac{t_n-1}{n} \frac{1}{n-1}.
\]
Note that if $k > n-t_n+1$ then
\[
\sum_{i=t_n}^{n-k+1}\frac{\binom{n-k}{i-1}}{\binom{n-1}{i-1}}\frac{1}{i-1} = 0
\]
and by Fact \ref{fact:threshold}
\[
r_{k,n} = \frac{t_n-1}{n} \frac{1}{n-1} \sim \frac{1}{e} \cdot \frac{1}{n}.
\]

Hence assume that $k \leqslant n-t_n+1$. Since $\frac{t_n-1}{n} \sim \frac{1}{e}$ and $\frac{t_n-1}{n} \frac{1}{n-1} \sim \frac{1}{e} \cdot \frac{1}{n}$, we need to show that $\sum_{i=t_n}^{n-k+1}\frac{\binom{n-k}{i-1}}{\binom{n-1}{i-1}}\frac{1}{i-1}$ is asymptotically smaller than $\frac{1}{n}$. Seeing that $\frac{\binom{n-k}{i-1}}{\binom{n-1}{i-1}} = \frac{\binom{n-i}{k-1}}{\binom{n-1}{k-1}}$ and that the function $f(i) = \frac{\binom{n-i}{k-1}}{\binom{n-1}{k-1}} \frac{1}{i-1}$ is decreasing in $i$ we have
\begin{equation} \label{ineq:sum}
\begin{split}
\sum_{i=t_n}^{n-k+1}\frac{\binom{n-k}{i-1}}{\binom{n-1}{i-1}}\frac{1}{i-1} & = \sum_{i=t_n}^{n-k+1}\frac{\binom{n-i}{k-1}}{\binom{n-1}{k-1}}\frac{1}{i-1} \leqslant \frac{n-t_n-k+2}{t_n-1} \frac{\binom{n-t_n}{k-1}}{\binom{n-1}{k-1}}\\
& = \frac{n-t_n-k+2}{t_n-1} \frac{n}{n-t_n+1} \frac{\binom{n-t_n+1}{k}}{\binom{n}{k}}.
\end{split}
\end{equation}
Since $k(n) = c \cdot n$ for some constant $c$ we get
\[
\frac{n-t_n-k+2}{t_n-1} \frac{n}{n-t_n+1} \sim \frac{e(1-c)-1}{1-1/e}
\]
and at the same time
\begin{align*}
\frac{\binom{n-t_n+1}{k}}{\binom{n}{k}} & = \left(1-\frac{k}{n-t_n+2}\right) \ldots \left(1-\frac{k}{n}\right) \leqslant  \left(1-\frac{k}{n}\right)^{t-1} \\
& \sim (1-c)^{n/e-1} = o(1/n).
\end{align*}
\qed
\end{proof}


\begin{lemma} \label{lemma:ratio_rnk_dec}
Let $r_{k,n}$ be defined as in Theorem \ref{thm:kth_cand}. For $k \in [n-2]$
\[
\lim_{n \to \infty} \frac{r_{k,n}}{r_{k+1,n}} \geqslant \lim_{n \to \infty} \frac{r_{k+1,n}}{r_{k+2,n}}.
\]
\end{lemma}
\begin{proof}
First note that when $k=k(n)$ is a function linear in $n$, by Lemma \ref{lemma:k_linear} we obtain
\[
\lim_{n \to \infty} \frac{r_{k,n}}{r_{k+1,n}} = \lim_{n \to \infty} \frac{r_{k+1,n}}{r_{k+2,n}} =1.
\]
Hence assume that 
$k = k(n) = o(n)$. Let $a_k = \sum_{s=k}^{\infty}\frac{1}{s}\left( 1-\frac{1}{e}\right)^s$. By Theorem \ref{thm:kth_cand} we can write
\[
\frac{r_{k,n}}{r_{k+1,n}} \sim \frac{a_k}{a_{k+1}} \quad \textnormal{and} \quad \frac{r_{k+1,n}}{r_{k+2,n}} \sim \frac{a_{k+1}}{a_{k+2}}.
\]
Let $\beta_k = \frac{1}{k}\left(1-\frac{1}{e}\right)^k$. Thus we have to prove
$
\frac{a_k}{a_k - \beta_k} \geqslant \frac{a_k - \beta_k}{a_k - \beta_k - \beta_{k+1}}
$
which is equivalent to
$
a_k \geqslant \frac{\beta_k^2}{\beta_k-\beta_{k+1}}.
$
Let $f(k) = \frac{\beta_k^2}{\beta_k-\beta_{k+1}}$. We need to show that
\begin{equation} \label{eq:}
\sum_{s=k}^{\infty}\frac{1}{s}\left( 1-\frac{1}{e}\right)^s \geqslant f(k).
\end{equation}
One can easily verify that for all $s \geqslant 1$
\[
\frac{1}{s}\left( 1-\frac{1}{e}\right)^s \geqslant f(s) - f(s+1).
\]
Summing both sides of the above inequality over $s \geqslant k$ we obtain
\[
\sum_{s=k}^{\infty}\frac{1}{s}\left( 1-\frac{1}{e}\right)^s \geqslant \sum_{s=k}^{\infty}(f(s)-f(s+1)) = f(k),
\]
where the last equality follows from the fact that $f(n) \xrightarrow{n \to \infty} 0$.
\qed
\end{proof}

\begin{theorem} \label{thm:ratio_rnk}
Let $l \geqslant 1$ be a constant and let $k \in \{2, \ldots, n-l\}$. Let $r_{k,n}$ be defined as in Theorem \ref{thm:kth_cand} and let also $a_l = \sum_{s=l}^{\infty} \frac{1}{s} \left( 1-\frac{1}{e}\right)^s$. Then
\[
\lim_{n \to \infty} \frac{r_{k,n}}{r_{k+l,n}} < \frac{1}{a_{l+1}}.
\]
\end{theorem}
\begin{proof}
By Lemma \ref{lemma:ratio_rnk_dec} we have
\begin{align*}
\lim_{n \to \infty} \frac{r_{k+1,n}}{r_{k+l,n}} & = \lim_{n \to \infty} \frac{r_{k+1,n}}{r_{k+2,n}} \cdot \lim_{n \to \infty} \frac{r_{k+2,n}}{r_{k+3,n}} \ldots \lim_{n \to \infty} \frac{r_{k+l-1,n}}{r_{k+l,n}} \\
& \leqslant \lim_{n \to \infty} \frac{r_{2,n}}{r_{3,n}} \cdot \lim_{n \to \infty} \frac{r_{3,n}}{r_{4,n}} \ldots \lim_{n \to \infty} \frac{r_{l,n}}{r_{l+1,n}} = \lim_{n \to \infty} \frac{r_{2,n}}{r_{l+1,n}}.
\end{align*}
Additionally, for $k \geqslant 2$, by Theorem \ref{thm:kth_cand} and Lemma \ref{lemma:ratio_rnk_dec}
\[
\lim_{n \to \infty} \frac{r_{k,n}}{r_{k+1,n}}  < \lim_{n \to \infty} \frac{r_{1,n}}{r_{2,n}}.
\]
Together, by Theorem \ref{thm:kth_cand}, it gives
\begin{align*}
\lim_{n \to \infty} \frac{r_{k,n}}{r_{k+l,n}} & = \lim_{n \to \infty} \frac{r_{k,n}}{r_{k+1,n}} \cdot \lim_{n \to \infty} \frac{r_{k+1,n}}{r_{k+l,n}} \\
& < \lim_{n \to \infty} \frac{r_{1,n}}{r_{2,n}} \lim_{n \to \infty} \frac{r_{2,n}}{r_{l+1,n}} = \lim_{n \to \infty} \frac{r_{1,n}}{r_{l+1,n}} = \frac{1}{a_{l+1}}.
\end{align*}
\qed
\end{proof}


\end{document}